\documentclass[11pt]{article}
\usepackage[utf8]{inputenc}
\usepackage[T1]{fontenc}
\usepackage{tikz}
\usetikzlibrary{shapes.geometric, arrows}
\usepackage{fullpage}

\usepackage{amsmath}
\usepackage{amssymb}
\usepackage{amsthm}
\usepackage{authblk}
\usepackage{comment}
\usepackage{subcaption}

\usepackage{xfrac}
\usepackage{algpseudocode}
\usepackage{algorithm}
\usepackage{tikz}

\usetikzlibrary{positioning}
\usepackage{xcolor}

\usepackage{hyperref}
\usepackage{cleveref}% http://ctan.org/pkg/cleveref
\usepackage{lipsum}% http://ctan.org/pkg/lipsum

\title{A Critique of Deng's ``P=NP''\thanks{Supported in part by NSF grant 
		CCF-2006496.}}

\author{Isabel Humphreys}
\author{Matthew Iceland}
\author{Harry Liuson}
\author{Dylan McKellips}
\author{Leo Sciortino}
\affil{Department of Computer Science\\University of Rochester\\Rochester, NY 14627, USA}

\newtheorem{theorem}{Theorem}
\newtheorem{lemma}{Lemma}

\newtheorem{claim}{Claim}
\newtheorem{definition}{Definition}

\date{July 11, 2025}

\begin{document}\sloppy

\maketitle

\begin{abstract}
In this paper, we critically examine Deng's ``P=NP'' \cite{deng2024pnp}. The paper claims that there is a polynomial-time algorithm that decides 3-coloring for graphs with vertices of degree at most 4, which is known to be an NP-complete problem. Deng presents a semidefinite program with an objective function that is unboundedly negative if the graph is not 3-colorable, and a minimum of 0 if the graph is 3-colorable. Through detailed analysis, we find that Deng conflates subgraphs with induced subgraphs, leading to a critical error which thereby invalidates Deng's proof that $\text{P}=\text{NP}$.
\end{abstract}

\section{Introduction}

The $\text{P} = \text{NP}$ problem is widely viewed as the most important problem in theoretical computer science. Essentially, the $\text{P} = \text{NP}$ problem asks whether all problems whose solutions can be verified efficiently (i.e., in polynomial time) can also be solved efficiently. One way to prove that $\text{P} = \text{NP}$ would be to give a polynomial-time algorithm for one of the many known NP-complete problems. However, no such algorithm has been discovered.

Holyer \cite{doi:10.1137/0210055} showed that determining whether a graph is 3-colorable is NP-complete, even when the degree of each vertex is at most 4. Deng proposes a polynomial-time reduction from such a graph to a certain convex optimization problem in Deng's Theorem 1.1 \cite{deng2024pnp}. The convex optimization problem can be solved in polynomial time. Since the 3-coloring problem is NP-complete, if the reduction were valid, all problems in NP could then be polynomial-time reduced to the semidefinite programming problem, proving that $\text{P} = \text{NP}$.

Specifically, Deng claims that a given graph with vertices of degree at most 4 is 3-colorable if and only if the semidefinite program of Theorem 1.1 is bounded, meaning that if the program is unboundedly negative, then the graph is not 3-colorable. If valid, Deng's work would resolve one of the most important open questions in computer science, with profound implications for fields ranging from cryptography to algorithm design. However, we find a critical fault in the use of lemma 3 to prove Theorem 1.1, which invalidates the proof's correctness.

In Section \ref{preliminaries}, we provide necessary background information. In Section \ref{dengsarguments}, we outline Deng's arguments, and finally in Section \ref{critique}, we give our critique.

\section{Preliminaries} \label{preliminaries}
In this paper, we expect a basic understanding of the complexity classes P and NP\@. Although not specifically needed to understand the paper, this provides context for the critique. The reader can obtain a comprehensive overview in Sipser's textbook \cite{sipser13}.

We say a graph is specified as $G=(V, E)$, where $V$ is a set of $n$ vertices $\{v_1, \dots, v_n\}$ and $E$ is a set of edges $\{(v_i, v_j) | v_i, v_j \in V\}$ \cite{sipser13}. We define the language $\text{3COLOR} = \{\langle G \rangle \ | \ G \text{ is colorable with 3 colors}\}$. A graph $G = (V, E)$ is 3-colorable if and only if there is a function $\phi: V \rightarrow  \{\text{red}, \text{green} , \text{blue}\}$ such that $\phi(u) \neq \phi(v)$ for all pairs of adjacent vertices (vertices such that $(u, v) \in E$). The problem of deciding whether a graph is 3-colorable is known to be NP-complete \cite{doi:10.1137/0210055}. Furthermore, restricting this problem to the set of graphs with vertex degree (number of edges attached to a given vertex) at most 4 is also known to be NP-complete, and it is this problem that Deng claims to solve. 

\subsection{Semidefinite matrices}
For a matrix $X$, we say $X \succeq 0$ if and only if $X$ is symmetric and positive semidefinite. We say a matrix $M\in \mathbb{R}^{n \times n}$  is positive semidefinite if and only if for all $x \in \mathbb{R}^n$,  $x^TMx \geq 0$. Equivalently, $M\in \mathbb{R}^{n \times n}$ is positive semidefinite if and only if all eigenvalues are nonnegative. 

\subsection{Semidefinite Programming}
Programs of the following form are said to be semidefinite programming problems \cite{freund2004sdp}: \\

Let $C$ and $X$ be symmetric $n \times n$ matrices, and let $$C \bullet X = \sum_{i=1}^{n}\sum_{j=1}^n C_{i, j}X_{i, j}.$$
A semidefinite program (SDP) optimization problem is of the following form.

\begin{align*}
    \textit{minimize}\quad &C \bullet X  \\
    \text{such that} \quad &A_i \bullet X = b_i, \quad 1\leq i \leq m \\
    &X \succeq 0.
\end{align*}
Each $A_i$ is an $n \times n $ symmetric matrix, and $A_i \bullet X$ forms a linear constraint to $b_i \in \mathbb{R}$. Therefore, there are $m$ linear constraints. We say that $X$ is the variable matrix, as $C$ and $A_i$ are fixed, and we are trying to find $X$ such that the objective function $C \bullet X$ is minimized with $X \succeq 0$.

Given this program, the corresponding dual is 
\begin{align*}
\textit{maximize} \quad & \sum y_i b_i \\
\text{such that} \quad & \sum (y_i A_i) + S = C\\
& S \succeq 0,
\end{align*}
where $y_i$ are variables and $S$ is a matrix.

\section{Deng's Arguments} \label{dengsarguments}

\subsection{Overview}
\begin{figure}
    \centering
   \tikzstyle{startstop} = [rectangle, rounded corners, minimum width=2cm, minimum height=1cm,text centered, draw=black]
\tikzstyle{process} = [rectangle, minimum width=2cm, minimum height=1cm, text centered, draw=black]
\tikzstyle{arrow} = [thick,->,>=stealth]
\begin{tikzpicture}[node distance=2cm]
\node (lemma5) [process] {Lemma 5};
\node (thm11) [process, below of=lemma5, yshift=0cm] {Thm 1.1};
\node (lemma1) [process, left of=lemma5, xshift=-6cm, yshift=0cm] {Lemma 1};
\node (lemma2) [process, below of=lemma1] {Lemma 2};
\node (thm32) [process, below of=lemma2] {Thm 3.2};
\node (thm33) [process, left of=thm11, xshift=-1.5cm] {Thm 3.3};
\node (lemma4) [process, left of=lemma5, xshift=-1.5cm, yshift=0cm] {Lemma 4};
\node (lemma3) [process, right of=lemma5, xshift=1.5cm, yshift=0cm] {Lemma 3};

\node (thm31) [process, right of=thm11, xshift=3cm] {Thm 3.1};
\node (pnp) [process, below of=thm11] {$P = NP$};

% Arrows
\draw [arrow] (lemma1) -- (lemma2);
\draw [arrow] (lemma2) -- (thm32);
\draw [arrow, <->] (thm33) -- (thm32) node[midway, above, sloped] {Equiv. by duality};
\draw [arrow] (lemma4) -- (thm11);
\draw [arrow] (lemma5) -- (thm11);
\draw [arrow] (lemma3) -- (thm11);
\draw [arrow] (thm33) to (thm11);
\draw [arrow, <->] (thm31) -- (thm11) node[midway, above, sloped] {Equiv. by duality};
\draw [arrow] (thm11) -- (pnp);

\end{tikzpicture}
    \caption{Relationship between Deng's theorems and lemmas}
    \label{fig:proof-diagram}
\end{figure}
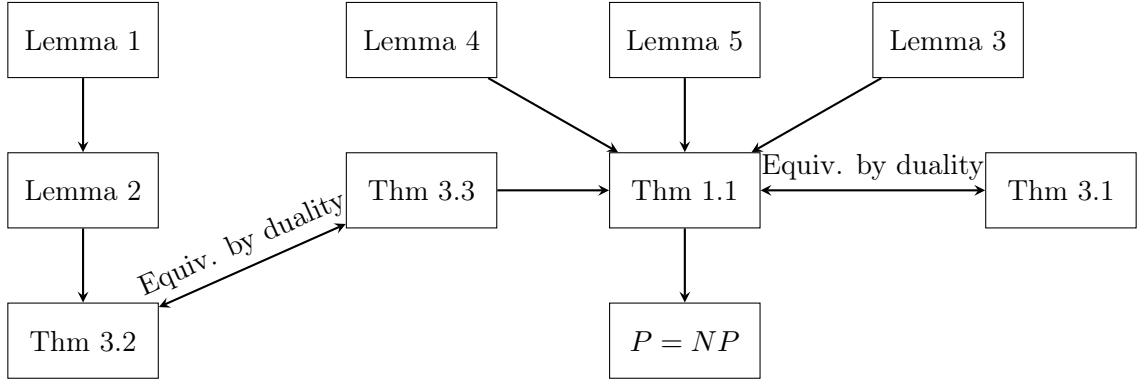

Figure \ref{fig:proof-diagram} illustrates the relationships between the theorems and lemmas of Deng's paper. For any two nodes $a$ and $b$ in the graph, if there is a directed edge from node $a$ to node $b$, then $a$ is necessary to prove $b$. The correctness of a node can be established only if all incoming nodes are correct.

Each of Theorem 1.1, Theorem 3.1, Theorem 3.2, and Theorem 3.3 gives a programming problem, including an objective function and a set of constraints. Theorems 1.1 and 3.1 are equivalent by duality theory, and Theorems 3.2 and 3.3 are likewise equivalent. For reference throughout our paper, Deng uses $R(G)$ to denote the programming problem of Theorem 1.1. As we describe in the following sections, Lemma 3 is used erroneously to prove one of the claims of Theorem 1.1.

\subsection{Theorem 1.1}
We begin with the central construction of the paper. Consider a graph $G = (V,E)$, of maximal degree 4. Let $n = \|V\|$.
Construct two relevant matrices: $D, P \in M_{3n+1}(\mathbb{R})$. Each is divided, excluding the last row and column, into $n^2$ submatrices of size $(3 \times 3)$ that will be referred to  as $D_{i,j}$ and $P_{i,j}$. The last row and column are defined as $(1 \times 3)$ submatrices for the row, and $(3 \times 1)$ submatrices for the column, with the shared corner being a $(1 \times 1)$ submatrix. 
Let the constraints $\Phi(G)$ be as follows:
\begin{enumerate}
    \item $D$ and $P$ are symmetric.
    \item $P$ has nonpositive values.
    \item $P + D \succeq 0$. 
    \item For all $i,j$, if $(v_i, v_j) \not \in E$, all nine entries in $D_{i,j}$ will be equal. 
    \item For all $i,j$, if $(v_i, v_j) \in E$, all nine entries in $P_{i,j} = 0.$
\end{enumerate}

Define the function $f(D) = C \bullet D$ where $$C_{i,j} = \left\{\begin{array}{lr}
     6 & \text{if }i = 3n+1 \text { AND } j = 3n+1\\
     4 & \text{if }i = 3n+1 \text{ XOR }j = 3n+1\\
     1 & \text{otherwise.}
\end{array}\right.$$

\begin{definition}
The Deng-semi-definite programming problem for graph $G$ is defined as $\min_{D,P} f(D)$ subject to the constraints $\Phi(G).$
\end{definition}

Deng claims that his program decides 3-coloring for graphs of vertex degree 4 or less in polynomial time. He claims to prove the following:
\begin{claim}
    \label{dsuff}
    If a graph $G$ is 3-colorable, then  $ \min_{D, P} f(D) = 0$.
\end{claim}

\begin{claim}
    \label{dnecc}
    If a graph $G$ is not 3-colorable, then $f(D)$ is unboundedly negative.
\end{claim}
Combined, Claim~\ref{dsuff} and Claim~\ref{dnecc} imply that $\min f(D)$ is either zero or unboundedly low. This implies a large (indeed, infinite) separation between the minima of $f$ for 3-colorable and non-3-colorable graphs. We find no issue with Claim~\ref{dsuff} and believe it to be correct. However, we take issue with the proof of Claim~\ref{dnecc}, which we will elaborate on below. The most impactful error arises from the use of Deng's Lemma 3 to prove Theorem 1.1.

\subsection{Convexity}
We note that the program defined in Theorem 1.1 of Deng's paper is not actually a semidefinite program, because it includes the constraint that the sum of two matrices is positive semidefinite, and as far as these authors are aware, there is no way of representing the ``Deng SDP'' in the common formulation shown in Section \ref{preliminaries}. However, since it is expressed purely in terms of linear constraints and linear matrix inequalities, it is nonetheless a convex optimization problem, for which there are known polynomial-time algorithms. This does not significantly change the analysis, but it is worth noting.

\subsection{Theorems 3.2 and 3.3}

We now provide a brief overview of Theorems 3.2 and 3.3 of Deng's paper, which involve the construction of copositive and completely positive programs. Definitions of copositive and completely positive matrices are given below.

\begin{definition}
    A copositive program  has the following form:
    \begin{align*}
    \textit{minimize}\quad &C \bullet X  \\
    \text{such that} \quad &A_i \bullet X = b_i, \quad 1\leq i \leq m \\
    &X \in \mathcal{C},
\end{align*}
where $\mathcal{C}$ is the set of copositive matrices defined below, in which $S$ is the set of all symmetric matrices: \begin{equation*}
    \mathcal{C} =  \{ M \in S \; |\; (\forall x \in \mathbb{R}^{n}_{\ge0})\; [xMx^T \geq 0] \}.
\end{equation*}
\end{definition}

\begin{definition}
    A completely positive program has the form:
    \begin{align*}
        \textit{minimize}\quad &C \bullet X  \\
        \text{such that} \quad &A_i \bullet X = b_i, \quad 1\leq i \leq m \\
        &X \in \mathcal{CP}.
    \end{align*}
    Where $\mathcal{CP}$ is the set of completely positive matrices defined below, in which $S$ is the set of all symmetric matrices: \begin{equation*}
    \mathcal{CP} =  \{ M \in S \; |\; (\forall D \in \mathbb{R}^{n \times n}_{\ge0})\; [M=D^{T}D] \}
    \end{equation*}
\end{definition}

Deng constructs a copositive program in Theorem 3.2 defined as $K(G)$ for some graph $G$, and claims that if $G$ is a D-graph, $K(G)$ will have an unbounded value. He then constructs its dual completely positive program $K^*(G)$ in Theorem 3.3, which states that any D-graph will not be able to satisfy the constraints of $K^*(G)$. Deng uses these programs to derive a necessary contradiction in his proof of Theorem 1.1.

\section{Critique of Deng's Lemma 3} \label{critique}

\subsection{Relevant definitions}

For a graph $G = (V,E)$, we define a \emph{subgraph} $G' = (V', E')$ to be a graph such that $E'\subseteq E$ and $V' \subseteq V$ and if $u,v \in V - V'$ then $(u,v) \notin E'$. For a graph $G = (V,E)$, we define a \emph{vertex-induced subgraph} $G' = (V', E')$ to be a graph such that $V' \subseteq V$ and $E = \{(u,v) \in E \mid u,v \in V'\}$. That is, all edges are present except for those which connect to one or more removed vertices. 

A proper coloring of a graph is a function $\phi : V \to \{0, \dots, k-1\}$  such that if $(u,v) \in E$ we have $\phi(u) \neq \phi(v)$. Furthermore, based on the parameter $k$ we call this a $k$-coloring. From now on, when we state \emph{coloring} we mean \emph{proper coloring}. The chromatic number of a graph is the smallest number of colors needed to color the vertices of a graph such that no two adjacent vertices have the same color.

We define a graph $(V, E)$ to be \emph{4-edge-critical} when it has chromatic number 4, and for any $E' \subsetneq E$, the edge-induced subgraph $(V, E')$ has lesser chromatic number. We define a graph $(V, E)$ to be \emph{4-vertex-critical} when it has chromatic number 4, and for any $V' \subsetneq V$, the vertex-induced subgraph $(V', E')$, with $E' = \{(u,v) \in E \mid u,v \in V'\}$, has lesser chromatic number. Deng defines a \emph{D-graph} (defined in Theorem 3.2) to be a 4-edge-critical graph with maximum degree 4 or less. Since we are only considering graphs with maximum degree 4 or less. For the remainder of this section we use the term D-graph interchangeably with 4-edge-critical.

The corresponding dual program to $R(G)$ will be referred to as $R^*(G)$, and $Z(G)$ is the matrix satisfying the constraints of $R^*(G)$ as defined in Theorem 3.1.

\subsection{Deng's Lemma 3}
Deng proposes Lemma 3, which we restate as follows.
\begin{claim}[\cite{deng2024pnp}, Lemma 3]
\label{deng-lemma-3}
For any graph G of chromatic number greater than 3, and maximum degree four, there exists a vertex-induced subgraph \footnote{In Lemma 3 of Deng's paper, $K$ is referred to as merely a subgraph of $G$. However, from Deng's proof of Lemma 3, it can be inferred that $K$ is a vertex-induced subgraph.} K that is a D-graph. If $R^*(K)$ does not have a constraint-satisfying solution, then neither does $R^*(G)$. 
\end{claim}

This statement is incorrect. The error lies in the fact that Deng implicitly describes a vertex-induced subgraph in his proof by using selected indices of the adjacency matrix to construct $K$. Constructing a general subgraph would involve changing values of the adjacency matrix instead. It is true that any non-3-colorable graph with degree at most four has a subgraph that is a D-graph (to see this, remove all non-critical edges until you are left with only critical edges), and it is true that every 4-colorable graph has a vertex induced subgraph that is 4-vertex-critical (to see this, remove non-critical vertices until we are only left with critical vertices), but it is not the case that every non-3-colorable graph has a  \textit{vertex-induced} subgraph that is a D-graph.

To produce a counterexample, take a D-graph with 2 or more vertices of degree 3 or less, and add an edge between any pair of these vertices.  Then, this is not a D-graph, since you have added an edge that is not critical (since we can remove it and will be left with a D-graph). Furthermore, it does not have any vertex-induced subgraph that is a D-graph, since if you were to remove any vertex (that does not disconnect the graph), you would be left with a 3-colorable graph (this follows from the definition of D-graph), which is not a D-graph. Note that adding such an edge will leave us with a 4-vertex-critical graph.

We can use the same construction for 4-vertex-critical graphs containing 2 or more vertices of degree 3 or less, which may or may not be D-graphs. These cannot have D-graphs as vertex-induced subgraphs, as any vertex induced subgraph is 3 colorable. If such a 4-vertex-critical graph is a D-graph, use the prior construction. If not, then we are done, since they are not D-graphs nor do they have a D-graph as a vertex-induced subgraph.

Consider Figure \ref{fig:moser} to illustrate this point. We start with a D-graph, the Moser spindle, and add the highlighted red edge. Then, this graph is not a D-graph, since if we remove the red edge it is not 3-colorable. Furthermore, removing any vertex also does not yield a D-graph, as removing any vertex will induce 3-colorability.

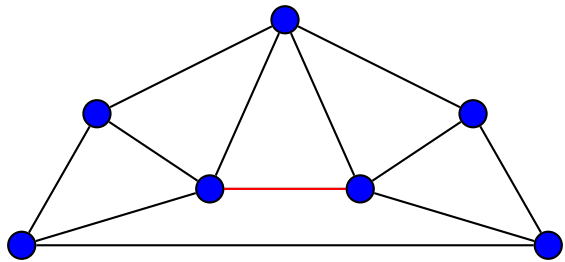
\begin{figure}[htp]
    \centering
            \begin{tikzpicture}[>=stealth, node distance=1.2cm and 1.5cm, every node/.style={circle, draw, minimum size=0.3cm}, thick]

  % Nodes
  \node[fill=blue] (1) at (0,3) {};
  \node[fill=blue] (2) at (-2.5,1.75) {};
  \node[fill=blue] (3) at (2.5,1.75) {};
  \node[fill=blue] (4) at (-3.5,0) {};
  \node[fill=blue] (5) at (3.5,0) {};
  \node[fill=blue] (6) at (-1,0.75) {};
  \node[fill=blue] (7) at (1,0.75) {};

  % Edges for main moser spindle graph
  \draw (1) -- (2);
  \draw (1) -- (3);
  \draw (1) to  (6);
  \draw (1) -- (7);
  \draw (6) -- (7);
  \draw (2) -- (4);
  \draw (3) -- (5);
  \draw (4) -- (5);
  \draw (2) to (6);
  \draw (3) -- (7);
  \draw (4) -- (6);
  \draw (5) -- (7);
  %red edge added to moser spindle
  \draw[red] (6) -- (7);

\end{tikzpicture}
    \caption{Moser spindle with the added red edge}
    \label{fig:moser}
\end{figure}

We wish to create an infinite family of such graphs. The following construction is by Thomas and Walls \cite{TWalls}. Consider the following series of graphs, shown in figure \ref{fig:F_graphs}.
Let $F_0 = K_4$, the complete graph with 4 vertices. Informally, we will construct subsequent iterations by ``replacing'' edges with the kite graph (\ref{fig:kite} \cite{Brandstadt04}). Let $F_1$ be constructed as follows:
\begin{enumerate}

    \item Remove any edge $(u_i,u_j)$ in $F_0$ shown in in figure \ref{fig:graph0}.
    \item Add new vertices $v_2$, $v_3$, and $v_4$ shown in figure \ref{fig:kite} which are all connected to each other.
    \item Add edges $(u_i,v_2)$, $(u_j,v_3)$, and $(u_j, v_4)$; The resulting graph is shown in figure \ref{fig:graph1}
\end{enumerate}
To create $F_{n+1}$ from $F_n$, we replace the edge shown as $(v_3, v_4)$ in figure \ref{fig:graph1} with another kite graph by repeating steps 2 and 3 with the equivalent vertices from $F_1$.

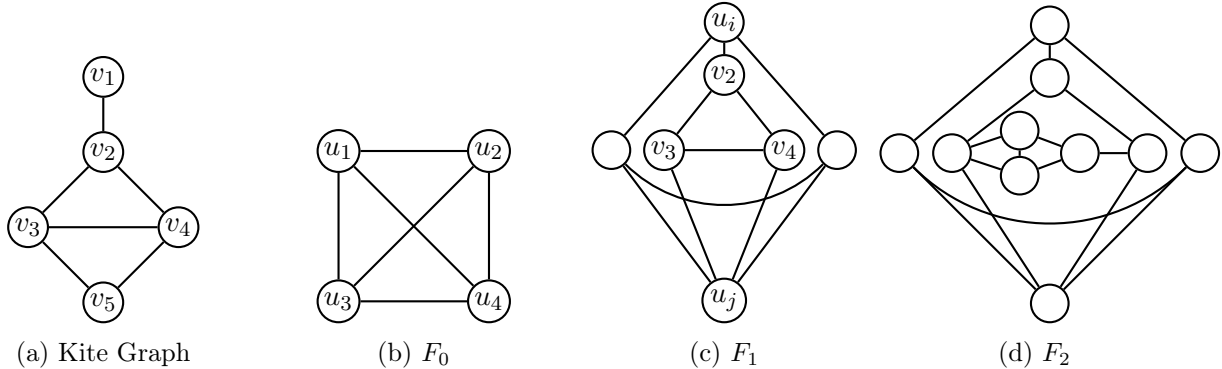
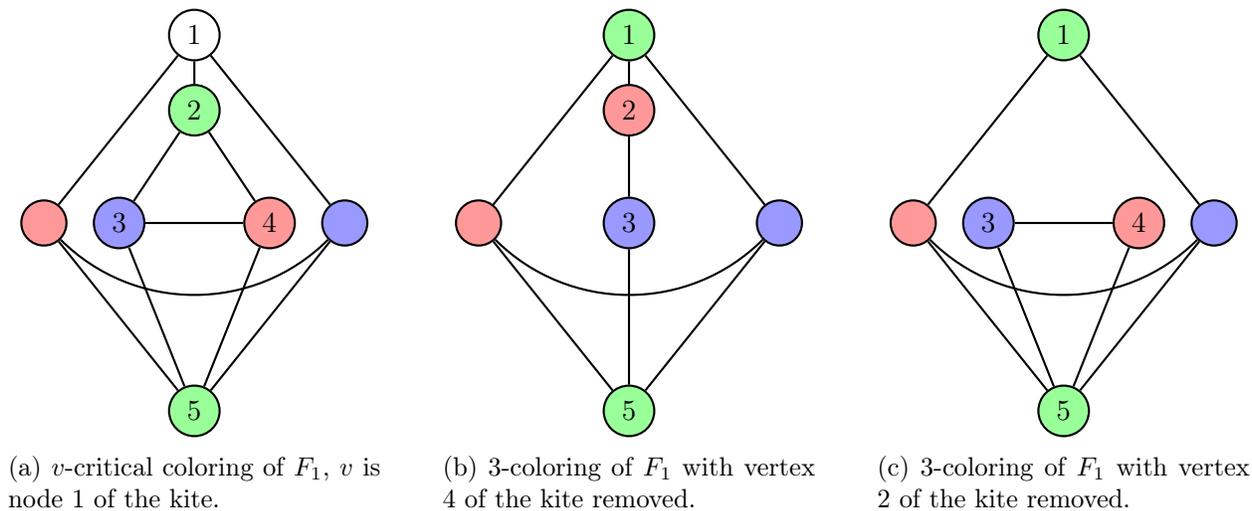
\begin{figure}[htp!]
    \centering
    \begin{subfigure}[b]{0.25\textwidth}\label{kite}
        \centering
        \begin{tikzpicture}[>=stealth, node distance=1.2cm and 1.5cm, every node/.style={circle, draw, minimum size=0.5cm, inner sep=0.04cm}, thick]

        \node (1) at (0, 3) {$v_1$};
        \node (2) at (0, 2) {$v_2$};
        \node (3) at (-1, 1) {$v_3$};
        \node (4) at (1, 1) {$v_4$};
        \node (5) at (0, 0) {$v_5$};

        \draw (1) to (2);
        \draw (2) to (4);
        \draw (2) to (3);
        \draw (3) to (4);
        \draw (3) to (5);
        \draw (4) to (5);
        \end{tikzpicture}
        \caption{Kite Graph}
        \label{fig:kite}
    \end{subfigure}\hfill
    \begin{subfigure}[b]{0.25\textwidth}
        \centering
        \begin{tikzpicture}[>=stealth, node distance=1.2cm and 1.5cm, every node/.style={circle, draw, minimum size=0.5cm, inner sep=0.04cm}, thick]
        \node (6) at (-1, 1) {$u_1$};
        \node (7) at (1, 1) {$u_2$};
        \node (8) at (-1, -1) {$u_3$};
        \node (9) at (1, -1) {$u_4$};

        \draw (6) to (7);
        \draw (6) to (8);
        \draw (8) to (9);
        \draw (9) to (7);
        \draw (6) to (9);
        \draw (7) to (8);

        \end{tikzpicture}
        \caption{$F_0$}
        \label{fig:graph0}
    \end{subfigure}\hfill
    \begin{subfigure}[b]{0.25\textwidth}
        \centering
       \begin{tikzpicture}[>=stealth, node distance=1.2cm and 1.5cm, every node/.style={circle, draw, minimum size=0.5cm, inner sep=0.04cm}, thick]

\node (1) at (0,3.7) {$u_i$};
  \node (2) at (-1.5,2) {};
  \node (3) at (1.5,2) {};
  \node (4) at (0,0) {$u_j$};

  % Kite middle layer
  \node (5) at (0,3) {$v_2$};
  \node (6) at (-.8,2) {$v_3$};
  \node (7) at (.8,2) {$v_4$};

  % Edges for main graph
  \draw (1) -- (2);
  \draw (1) -- (3);
  \draw (2) to[bend right=45] (3);
  \draw (2) -- (4);
  \draw (3) -- (4);

  % Edges for kite graph
  \draw (1) -- (5);
  \draw (5) -- (6);
  \draw (5) -- (7);
  \draw (6) -- (7);
  \draw (6) -- (4);
  \draw (7) -- (4);

\end{tikzpicture}
        \caption{$F_1$}
        \label{fig:graph1}
    \end{subfigure}\hfill
    \begin{subfigure}[b]{0.25\textwidth}
        \centering
          \begin{tikzpicture}[>=stealth, node distance=1.2cm and 1.5cm, every node/.style={circle, draw, minimum size=0.5cm, inner sep=0.04cm}, thick]

\node (1) at (0,3.7) {};
  \node (2) at (-2,2) {};
  \node (3) at (2,2) {};
  \node (4) at (0,0) {};

  % Kite middle layer
  \node (5) at (0,3) {};
  \node (6) at (-1.3,2) {};
  \node (7) at (1.3,2) {};

    % Inner Inner kite layer
  \node (8) at (.4, 2) {};
  \node (9) at (-.4, 2.3) {};
  \node (10) at (-.4, 1.7) {};

  % Edges for main graph
  \draw (1) -- (2);
  \draw (1) -- (3);
  \draw (2) to[bend right=45] (3);
  \draw (2) -- (4);
  \draw (3) -- (4);

  % Edges for kite graph
  \draw (1) -- (5);
  \draw (5) -- (6);
  \draw (5) -- (7);
  \draw (6) -- (4);
  \draw (7) -- (4);

  % Inner kite edges
  \draw (6) -- (9);
    \draw (6) --(10);
    \draw (9) -- (10);
    \draw (9) to (8);
    \draw (10) to (8);
    \draw (8) to (7);

\end{tikzpicture}
        \caption{$F_2$}
        \label{fig:houghspace}
    \end{subfigure}
    \caption{Graphs where the central edge is replaced by the kite graph}
    \label{fig:F_graphs}
\end{figure}

These graphs are all 4-vertex-critical \cite{4critproof}, so we know that they do not have subgraphs that are D-graphs. They themselves may be D-graphs, however we can account for this possibility. Observe that $F_n$ has $n$ vertices with degree 4, and each graph has $4+3n$ vertices, leaving $4+2n$ vertices of degree 3. This is because whenever we add a kite, we change one vertex from degree 3 to degree 4 and add three new vertices of degree 3. Thus, we can freely add an edge to a pair of degree 3 vertices without inducing a D-graph as a subgraph (as adding edges will keep the 4-vertex-critical property), and by doing this we eliminate the possibility of $F_n$ being a D-graph.

Observe that for any vertex $v$ in a 4-vertex-critical graph $G$, $G$ can be 4-colored with the special criterion that one of the colors is used only once, to color $v$. To see this, remove $v$, 3-color the graph, and then add $v$ back and color it with the fourth color. Denote such a coloring a $v$-critical-coloring of $G$. Note that this is an equivalent characterization of vertex-criticality.
\begin{figure}[htp!]
    \centering
    \begin{subfigure}[b]{0.30\textwidth}
        \centering
        \begin{tikzpicture}[>=stealth, node distance=1.2cm and 1.5cm, every node/.style={circle, draw, minimum size=0.6cm}, thick]

  % Nodes
  \node[fill=white!40] (1) at (0,5) {1};
  \node[fill=red!40] (2) at (-2,2.5) {};
  \node[fill=blue!40] (3) at (2,2.5) {};
  \node[fill=green!40] (4) at (0,0) {5};

  % Kite middle layer
  \node[fill=green!40] (5) at (0,4) {2};
  \node[fill=blue!40] (6) at (-1,2.5) {3};
  \node[fill=red!40] (7) at (1,2.5) {4};

  % Edges for main graph
  \draw (1) -- (2);
  \draw (1) -- (3);
  \draw (2) to[bend right=45] (3);
  \draw (2) -- (4);
  \draw (3) -- (4);

  % Edges for kite graph
  \draw (1) -- (5);
  \draw (5) -- (6);
  \draw (5) -- (7);
  \draw (6) -- (7);
  \draw (6) -- (4);
  \draw (7) -- (4);

\end{tikzpicture}
        \caption{$v$-critical coloring of $F_1$, $v$ is node 1 of the kite.}
        \label{fig:v-critical}
    \end{subfigure}\hfill
    \begin{subfigure}[b]{0.30\textwidth}
        \centering
                \begin{tikzpicture}[>=stealth, node distance=1.2cm and 1.5cm, every node/.style={circle, draw, minimum size=0.6cm}, thick]

  % Nodes
  \node[fill=green!40] (1) at (0,5) {1};
  \node[fill=red!40] (2) at (-2,2.5) {};
  \node[fill=blue!40] (3) at (2,2.5) {};
  \node[fill=green!40] (4) at (0,0) {5};

  % Kite middle layer
  \node[fill=red!40] (5) at (0,4) {2};
  \node[fill=blue!40] (6) at (0,2.5) {3};

  % Edges for main graph
  \draw (1) -- (2);
  \draw (1) -- (3);
  \draw (2) to[bend right=45] (3);
  \draw (2) -- (4);
  \draw (3) -- (4);

  % Edges for kite graph
  \draw (1) -- (5);
  \draw (5) -- (6);
  \draw (6) -- (4);

\end{tikzpicture}
        \caption{3-coloring of $F_1$ with vertex 4 of the kite removed.}
        \label{fig:3-color-no-v4}
    \end{subfigure}\hfill
    \begin{subfigure}[b]{0.3\textwidth}
        \centering
        \begin{tikzpicture}[>=stealth, node distance=1.2cm and 1.5cm, every node/.style={circle, draw, minimum size=0.6cm}, thick]

  % Nodes
  \node[fill=green!40] (1) at (0,5) {1};
  \node[fill=red!40] (2) at (-2,2.5) {};
  \node[fill=blue!40] (3) at (2,2.5) {};
  \node[fill=green!40] (4) at (0,0) {5};

  % Kite middle layer
  \node[fill=blue!40] (6) at (-1,2.5) {3};
  \node[fill=red!40] (7) at (1,2.5) {4};

  % Edges for main graph
  \draw (1) -- (2);
  \draw (1) -- (3);
  \draw (2) to[bend right=45] (3);
  \draw (2) -- (4);
  \draw (3) -- (4);

  % Edges for kite graph
  \draw (6) -- (7);
  \draw (6) -- (4);
  \draw (7) -- (4);

\end{tikzpicture}
        \caption{3-coloring of $F_1$ with vertex 2 of the kite removed.}
        \label{fig:3-color-no-v2}
    \end{subfigure}\hfill
    \caption{Graphs where the central edge is replaced by the kite graph}
    \label{fig:colorings}
\end{figure}
\begin{lemma}
Every graph $F_n$ has chromatic number four.
\label{lem-fn-4color}
\end{lemma}

\begin{proof}
    First, we must show that $F_{n}$ is not 3-colorable for any $n \geq 0$. Observe that in order for the kite replacement technique to induce 3-colorability in a non-3-colorable graph, it would have to replace a critical edge and allow the two nodes corresponding to that edge to be the same color as each other. Denote the two nodes as $u_i$ and $u_j$ to allow for reference to figure \ref{fig:graph1}. Then color the original graph such that ($u_i$,$u_j$) is the only edge with two same-colored endpoints. This is possible because $(u_i,u_j)$ is a critical edge. Let us say the colors are $C_1$, $C_2$, and $C_3$, and $u_i$ and $u_j$ are both $C_1$. Then $v_3$ and $v_4$ from the new kite must be $C_2$ and $C_3$, forcing $v_2$ to be $C_1$. Therefore the addition of a kite in place of a critical edge cannot induce 3-colorability because $(u_i,v_2)$ is the new critical edge that plays the same role as the original edge.
    
    Clearly $F_0$ is not 3-colorable. This can be easily verified by noticing that all four vertices must be different from the other three, and that is not possible with only three colors. Then using the above logic, it follows by induction that $F_n$ is not 3-colorable for any $n \geq 0$.

    Now, we prove by induction that all graphs $F_{n}$ are 4-colorable. Observe that $F_{0}$ shown in figure \ref{fig:graph0} is 4-colorable since it has four vertices. Now, assume that graph $F_{k-1}$ is 4-colorable. To prove that graph $F_{k}$ is 4-colorable, notice that the 3 added vertices of the new kite can be colored in such a way so as to preserve 4-colorability while leaving the other vertices unchanged. To see this, we make reference to Figure \ref{fig:graph1}, where the 3 added vertices are $v_{2}$, $v_{3}$, and $v_{4}$. Color vertex $v_{2}$ differently from vertex $v_{1}$, and color vertices $v_{3}$ and $v_{4}$ differently from vertices $v_{2}$ and $v_{5}$ and differently from each other. The resulting coloring of $F_{k}$ is a valid 4-coloring.
\end{proof}

\begin{theorem}
Every graph $F_n$ is 4-vertex-critical.
\end{theorem} 
\begin{proof}

The base case is $F_0$, which is trivially 4-vertex critical, as removing any vertex gives $K_3.$
Now, suppose that $F_{k-1}$ is 4-vertex-critical, and consider $F_{k}$. We claim the following:
\begin{itemize}
    \item All previous critical vertices (the entire graph) remain critical vertices.
    \item The added vertices of the kite graph (vertices 2, 3, and 4) are critical vertices.
\end{itemize}

From now on we will refer to vertices of the kite using the indices shown in Figure \ref{fig:colorings}. We begin with the first claim. Let $v$ be a vertex of $F_{k-1}$ and consider the $v$-critical coloring of $F_{k-1}$. Then we can furnish a $v$-critical coloring of $F_{k}$ in a straightforward manner: vertex 2 must have the same color as vertex 5, and vertices 3 and 4 are each one of the other two colors.

To prove the second claim, consider a $v$-critical coloring $\phi$ where $v$ is vertex 1 of the kite graph inserted to make $F_k$ (Figure~\ref{fig:v-critical}).\footnote{Figure \ref{fig:colorings} depicts $F_1$, but we can use it to examine the inner subgraph of any $F_k$.} To show the vertex criticality of the added vertices, we furnish 3-colorings of the graph when they are removed. To construct a 3-coloring of $F_{k}$ with some arbitrary vertex removed from the 4-coloring $\phi$ of $F_{k-1}$, we need to recolor the vertex with the 4th color (vertex 1 of the kite) as well as any added vertices.

We begin by coloring vertex 1, which always has exactly two neighbors external to the kite graph. We color it accordingly, avoiding edges that connect two vertices of the same color. If vertex 4 is removed, vertices 2 and 3 are now of degree 2, so coloring them is trivial, and we have a 3-coloring of $F_n$ (Figure~\ref{fig:3-color-no-v4}). This proof applies if vertex 3 were removed instead, so both vertices 3 and 4 of the kite are critical vertices.

Suppose instead that vertex 2 of the kite is removed. Then vertices 1, 3, and 4 are all of degree 2 and can be colored trivially as shown in Figure~\ref{fig:3-color-no-v2}. Thus vertex 2 is critical. Therefore, all vertices of $F_k$ are critical. Then alongside Lemma~\ref{lem-fn-4color} this gives that $F_k$ is 4-vertex critical if $F_{k-1}$ is. By induction, the result is proven.
\end{proof}
Therefore there exists an infinite family of graphs of chromatic number four and maximum degree four such that no vertex-induced subgraph is a D-graph.

\subsection{Usage of Deng's Lemma 3}

Before describing how Lemma 3 is used in Deng's Theorem 1.1, we must define two terms Deng uses: D-coloring and D-coloring-matrix. A D-coloring with respect to an edge $e_{i, j}$ connecting the vertices $v_{i}$ and $v_{j}$ is a set of colorings for a D-graph such that it has no edge with two same-colored endpoints except $e_{i,j}$. By definition, every edge of a D-graph has a D-coloring.

To define a D-coloring matrix, let three variables $e_{i}$, $f_{i}$, $g_{i}$ denote the color of a given vertex $v_{i}$. Coloring it red corresponds to $e_{i}=1$, $f_{i}=0$, $g_{i}=0$, yellow to $e_{i}=0$, $f_{i}=1$, $g_{i}=0$, and blue to $e_{i}=0$, $f_{i}=0$, $g_{i}=1$. There are 6 equivalent permutations of a coloring, and these permutations are represented by $e_{i}^{(k)}$, $f_{i}^{(k)}$, $g_{i}^{(k)}$ for $k \in \{1, \dots, 6\}$. Let $X = [X_{1} \  X_{2} \ \dots X_{n} \ 1]$ $X_{i} = [x_{i} \ y_{i} \ z_{i}]$ be a D-coloring with respect to the edge $e_{i, j}$. Then the corresponding D-coloring matrix $L$ with respect to edge $e_{i, j}$ is defined by

\[
\sum_{k=1}^{6} (1 + \sum_{i=1}^{n} (e_{i}^{(k)}x_{i} + f_{i}^{(k)}y_{i} + g_{i}^{(k)}z_{i}))^{2} = XLX^{T}.
\]

Lemma 3 is used to prove the necessity (i.e. Claim 2) of Theorem 1.1. Using the lemma, Deng selects a D-graph $K$ as a vertex-induced subgraph of G, and he notes that there exists a matrix $Z^{(0)}(K)$ satisfying the dual $R^*(G)$ of $R(G)$ as defined in Theorem 3.1. He states that for all vertices $v_{i}$ in $K$, we can obtain a D-coloring matrix $A_i$ for any edge associated with $v_i$. Furthermore, for each vertex $v_i$, Deng constructs a second D-coloring matrix $B_i$ by changing the color of only vertex $v_i$ while keeping all other vertex colors the same. Next, Deng constructs a D-coloring matrix $L_{e_j}$ for each edge.  Deng then composes these matrices as follows, where $s$ is the order of K: \begin{align}
    W &= \sum_{i=1}^s (a_i A_i + b_i B_i).\\
    Y &= \sum_{e_j \in E(K)} c_{e_j} L_{e_j}.
\end{align}

Since we are not guaranteed such a D-graph K, or an initial D-coloring, we are not guaranteed to be able to find such an $A_i$ for each vertex, a second D-coloring, or a D-coloring matrix for each edge.

Therefore we are not guaranteed the existence of W or Y. Thus, the rest of Deng's conclusions do not follow. We cannot construct a kernel of these matrices, and we cannot find coefficients such that $Z^{(0)}(K) + W + Y$ is completely positive, since none of these matrices are guaranteed to exist, and so we cannot take these matrices to find a solution to Theorem 3.3 via copositive programming, as Deng intends. This solution to Theorem 3.3 is needed to derive a necessary contradiction, as Theorem 3.3 states that for any D-graph $G$, no matrix exists that satisfies the constraints of the programming problem $K^{*}(G)$. Thus, Deng's proof fails, and necessity is not established.

\section{Conclusion}
Deng presents a convex optimization problem, and we find, interestingly, that Claim~\ref{dsuff} appears to hold. However, we find significant errors in the proof of Claim~\ref{dnecc}, specifically with Deng's construction and claims about D-graphs. Deng also fails to properly justify the construction of the dual to the Deng SDP, and he makes a number of other unjustified or poorly justified assertions. Therefore, Deng's work fails to demonstrate that $\text{P} = \text{NP}$.

\paragraph{Acknowledgments}
We would like to thank Michael Reidy, Nicholas DeJesse, Spencer Lyudovyk, Tran Duy Anh Le, and Lane Hemaspaandra 
for their helpful comments on prior drafts.
The authors are responsible for any remaining errors.

\bibliographystyle{alpha}
\bibliography{citations}
\end{document}